\documentclass[12pt,leqno,letterpaper]{article}

\usepackage{amsmath,amsthm,enumerate,amssymb}
\usepackage[utf8]{inputenc}

\usepackage[english]{babel}
\usepackage{graphicx}
\usepackage{color}

\usepackage{graphicx}

\newtheorem{theorem}{Theorem}[section]

\newtheorem{lemma}[theorem]{Lemma}
\newtheorem{corollary}[theorem]{Corollary}

\newcommand{\tr}{{\rm Tr\hskip -0.2em}~}

\DeclareMathOperator{\frechetdiff}{\mathit d}
\newcommand{\fd}[1]{\frechetdiff\hskip -0.3em{#1}}

          \newcommand{\df}[2]{\frac{d#1}{d#2}}
          
\begin{document}

\title{Golden-Thompson's inequality\\
for\\
deformed exponentials}
\author{Frank Hansen}
\date{September 3, 2014}

\maketitle

\begin{abstract} Deformed logarithms and their inverse functions, the deformed exponentials, are important tools in the theory of non-additive entropies and non-extensive statistical mechanics. We formulate and prove counterparts of  Golden-Thompson's trace inequality for
 $ q $-exponentials with parameter $ q $ in the interval $ [1,3]. $
\\[1ex]

\end{abstract}

\section{Introduction and main result}

Tsallis~\cite{kn:tsallis:2009} generalised in 1988 the standard Bolzmann-Gibbs entropy to a non-extensive quantity $ S_q $ depending on a parameter $ q. $ In the quantum version it is given by
\[
S_q(\rho)=\frac{1-\tr\rho^q}{q-1}\qquad q\ne 1,
\]
where $ \rho $ is a density matrix. It has the property that
$
S_q(\rho)\to S(\rho)
$
for $ q\to 1, $ where $ S(\rho)=-\tr\rho\log\rho $ is the von Neumann entropy. The Tsallis entropy may be written on a similar form
\[
S_q(\rho)=-\tr\rho\log_q(\rho),
\]
where the deformed logarithm $ \log_q $ is given by
\[
\log_q x=\int_1^x  t^{q-2}\,dt = \left\{\begin{array}{ll}
                                                         \displaystyle\frac{x^{q-1}-1}{q-1}\quad &q> 1\\[2ex]
                                                         \log x                                 &q=1
                                                         \end{array}\right.
\]
for $ x>0. $ The deformed logarithm is also denoted the $ q $-logarithm. The inverse function $ \exp_q $ is called the $ q $-exponential and is given by 
\[
\exp_q(x)=(x(q-1)+1)^{1/(q-1)}\qquad\text{for}\quad x>\frac{-1}{q-1}\,.
\]
The $ q $-logarithm and the $ q $-exponential functions converge, respectively, to the logarithmic and the exponential functions for $ q\to 1. $

The aim of this article is to generalise Golden-Thompson's trace inequality \cite{kn:golden:1965, kn:thompson:1965} to deformed exponentials. The main result is the following:

\begin{theorem}\label{main theorem} 

Let $ A $ and $ B $ be positive definite matrices.

\begin{enumerate}[(i)]

\item If $ 1\le q<2 $ then
\[
\tr\exp_q(A+B)\le \tr\exp_q(A)^{2-q}\bigl(A(q-1) + \exp_q B\big).
\]

\item If $ 2\le q \le 3 $ then
\[
\tr\exp_q(A+B)\ge \tr\exp_q(A)^{2-q}\bigl(A(q-1) + \exp_q B\big).
\]

\end{enumerate}
\end{theorem}

Notice that we for $ q=1 $ recovers Golden-Thomson's trace inequality
\[
\tr\exp(A+B)\le\tr\exp(A)\exp(B).
\]
This inequality is valid for arbitrary self-adjoint matrices $ A $ and $ B. $ However, it is sufficient to know the inequality for positive definite matrices, since the general form follows by multiplication with positive numbers.

\section{Preliminaries}

We collect a few well-known results that we are going to use in the proof of the main theorem.

The $ q $-logarithm is a bijection of the positive half-line onto the open interval $ (-(q-1)^{-1},\infty), $ and the $ q $-exponential is consequently a bijection of the interval $ (-(q-1)^{-1},\infty) $ onto the positive half-line. For $ q>1 $ we may thus safely apply both the
$ q$-logarithm and the $q $-exponential to positive definite operators.
We also notice that
\begin{equation}\label{derivative of $ q-exponential}
\df{}{x}\log_q(x)=x^{q-2}\qquad\text{and}\qquad \df{}{x}\exp_q(x)=\exp_q(x)^{2-q}\,.
\end{equation}
The proof of the following lemma is rather easy and may be found in \cite[Lemma 5]{kn:lieb:1973:1}.

\begin{lemma}\label{differential inequality}
Let $ \varphi\colon\mathcal D\to\mathcal A_\text{sa} $ be a map defined in a convex cone $ \mathcal D $ in a Banach space $ X $ with values in the self-adjoint part of a $ C^* $-algebra $ \mathcal A. $ If $ \varphi $  is Fréchet differentiable, convex and positively homogeneous then
\[
d\varphi(x)h\le \varphi(h).
\]
for $ x,h\in\mathcal D. $
\end{lemma}

Let $ H $ be any $ n\times n $ matrix. The map
\[
A\to \tr (H^* A^p H)^{1/p},
\]
defined in positive definite $ n\times n $ matrices,
is concave for $ 0<p\le 1 $ and convex for $ 1\le p\le 2, $ cf. \cite[Theorem 1.1]{kn:carlen:2008}.  By a slight modification of the construction given in Remark 3.2 in the same reference, cf. also \cite{kn:hansen:2014}, we obtain that the mapping
\begin{equation}\label{Carlen-Lieb result}
(A_1,\dots,A_k)\to\tr (H_1^*A_1^pA_1+\cdots+H_k^* A_k H_k)^{1/p},
\end{equation}
defined in $ k $-tuples of positive definite $ n\times n $ matrices, is concave for $ 0<p\le 1 $ and convex for $ 1\le p\le 2; $ for arbitrary $ n\times n $ matrices $ H_1,\dots,H_k. $

\section{Deformed trace functions}

\begin{theorem}\label{concavity (convexity) of deformed trace functions}
Let $ H_1,\dots,H_k $ be matrices with $ H_1^*H_1 +\cdots+ H_k^*H_k=1 $ and define the function
\begin{equation}\label{definition of varphi-function}
\varphi(A_1,\dots,A_k)=\tr\exp_q\Bigl(\sum_{i=1}^k H_i^*\log_q(A_i) H_i\Bigr)
\end{equation}
in $ k $-tuples of positive definite matrices.
Then $ \varphi $ is positively homogeneous of degree one. It is concave for $ 1\le q\le 2 $ and convex for $ 2\le q\le 3. $
\end{theorem}

\begin{proof}
For $ q>1 $ we obtain
\[
\begin{array}{l}
\varphi(A_1,\dots, A_k)
=\displaystyle\tr\exp_q\Bigl(\sum_{i=1}^k H_i^*\log_q(A_i) H_i\Bigr)\\[1ex]
=\displaystyle\tr\Bigl((q-1)\Bigl(\sum_{i=1}^k H_i^*\log_q(A_i)H_i\Bigr)+1\Bigr)^{1/(q-1}\\[1ex]
=\displaystyle\tr\Bigl((q-1)\Bigl(\sum_{i=1}^k H_i^*\frac{A_i^{q-1}-1}{q-1}H_i\Bigr)+1\Bigr)^{1/(q-1)}\\[1ex]
=\displaystyle\tr\Bigl(\sum_{i=1}^k H_i^*(A_i^{q-1} -1)H_i+1\Bigr)^{1/(q-1)}\\[3.5ex]
=\displaystyle\tr\bigl(H_1^* A_1^{q-1}H_1+\cdots+H_k^*A_k^{q-1}H_k\bigr)^{1/(q-1)}.
\end{array}
\]
From this identity it follows that $ \varphi $ is positively homogeneous of degree one. The concavity for $ 1<q\le 2 $ and the convexity for $ 2\le q\le 3 $ now follows from (\ref{Carlen-Lieb result}). The statement for $ q=1 $ follows by letting $ q $ tend to one. 
\end{proof}

\begin{corollary}
Let $ L $ be positive definite, and let $ H_1,\dots,H_k $ be matrices such that $ H_1^*H_1 +\cdots+ H_k^*H_k\le 1. $
Then the function
\[
\varphi(A_1,\dots,A_k)=\tr\exp_q\bigl(L+H_1^*\log_q(A_1)H_1+\cdots+H_k^*\log_q(A_k)H_k\bigr),
\]
defined in $ k $-tuples of positive definite matrices, is concave for $ 1\le q\le 2 $ and convex for $ 2\le q\le 3. $ 
\end{corollary}

\begin{proof} We may without loss of generality assume $ H_1^*H_1 +\cdots+ H_k^*H_k< 1 $ and put $ H_{k+1}=\bigl(1-( H_1^*H_1 +\cdots+ H_k^*H_k)\bigr)^{1/2}. $ We then have
\[
H_1^*H_1+\cdots+H_k^*H_k+H_{k+1}^*H_{k+1}=1
\]
and may use the preceding theorem to conclude that the function
\[
\begin{array}{l}
(A_1,\dots,A_{k+1})\to
\tr\exp_q\bigl(H_1^*\log_q(A_1)H_1+\cdots+H_{k+1}^*\log_q(A_{k+1})H_{k+1}\bigr)
\end{array}
\]
of $ k+1 $ variables is  concave for $ 1\le q\le 2 $ and convex for $ 2\le q\le 3. $ Since $ H_{k+1} $ is invertible we may choose
\[
A_{k+1}=\exp_q\bigl(H_{k+1}^{-1}LH_{k+1}^{-1}\bigr)
\]
which makes sense since $ H_{k+1}^{-1}LH_{k+1}^{-1} $ is positive definite. Concavity for $ 1\le q\le 2 $ and  convexity for $ 2\le q\le 3 $ in the first $ k $ variables of the above function then yields the result.
\end{proof}

Setting $ q=1 $ we recover in particular \cite[Theorem 3]{kn:lieb:2005}.

\begin{corollary}
Let $ H_1,\dots,H_k $ be matrices with $ H_1^*H_1+\cdots+H_k^*H_k\le 1, $ and let $ L $ be self-adjoint. The trace function
\[
(A_1,\dots,A_k)\to \tr\exp\bigl(L+H_1^*\log(A_1)H_1+\cdots+H_k^* \log (A_k)H_k\bigr)
\]
is concave in positive definite matrices.
\end{corollary}

\begin{corollary}\label{main inequality for k variables}
The trace function $ \varphi $ defined in (\ref{definition of varphi-function}) satisfies
\[
\varphi(B_1,\dots,B_k)\le
\displaystyle\tr\exp_q \Bigl(\sum_{i=1}^k H_i^*\log_q(A_i) H_i\Bigr)^{2-q} \sum_{j=1}^k H_j^*(\fd{}\log_q (A_j)B_j) H_j
\]
for $ 1\le q\le 2 $ and
\[
\varphi(B_1,\dots,B_k)\ge
\displaystyle\tr\exp_q \Bigl(\sum_{i=1}^k H_i^*\log_q(A_i) H_i\Bigr)^{2-q} \sum_{j=1}^k H_j^*(\fd{}\log_q(A_j)B_j) H_j
\]
for $ 2\le q\le 3, $ where $ A_1,\dots,A_k $ and $ B_1,\dots,B_k $ are positive definite matrices.
\end{corollary}

\begin{proof} For $ 1\le q\le 2 $ we obtain
\[
\fd{}\varphi(A_1,\dots,A_k)(B_1,\dots,B_k)\ge \varphi(B_1,\dots,B_k)
\]
by Lemma~ \ref{differential inequality}. By the chain rule for Fréchet differentiable mappings between Banach spaces we therefore obtain
\[
\begin{array}{l}
\varphi(B_1,\dots,B_k)\le\displaystyle\sum_{j=1}^k d_j \varphi(A_1,\dots,A_k)B_j \\[2ex]
=\displaystyle\sum_{j=1}^k \tr\fd{}\exp_q \Bigl(\sum_{i=1}^k H_i^*\log_q(A_i) H_i\Bigr) H_j^*(\fd{}\log_q(A_j)B_j) H_j\\[3ex]
=\displaystyle\sum_{j=1}^k \tr\exp_q \Bigl(\sum_{i=1}^k H_i^*\log_q(A_i) H_i\Bigr)^{2-q} H_j^*(\fd{}\log_q(A_j)B_j) H_j
\end{array}
\]
where we used the identity $ \tr\fd{}f(A)B=\tr f'(A)B $ valid for differentiable functions.
This proves the first assertion. The result for $ 2\le q\le 3 $ follows similarly.
\end{proof}

\section{Proof of the main theorem}

In order to prove Theorem~\ref{main theorem} (i) we set $ k=2 $ in Corollary~\ref{main inequality for k variables} and obtain 
\[
\varphi(B_1,B_2)\le\tr\exp_q(X)^{2-q}\bigl(H_1^* (\fd{}\log_q(A_1)B_1) H_1 + H_2^* (\fd{}\log_q(A_2)B_2) H_2\big)
\]
for $ 1\le q\le 2 $ and positive definite matrices $ A_1,A_2 $ and $ B_1,B_2 $ where
\[
X=H_1^*\log_q(A_1) H_1 + H_2^* \log_q(A_2) H_2\,.
\]
If we set $ A_1=B_1 $ and $ A_2=1 $ the inequality reduces to
\[
\varphi(B_1,B_2)\le \tr\exp_q(H_1^* \log_q(B_1) H_1)^{2-q}\bigl(H_1^* B_1^{q-1} H_1 +H_2^* B_2 H_2\bigr).
\]
We now set $ H_1=\varepsilon^{1/2} $ for $ 0<\varepsilon<1, $ and to fixed positive definite matrices $ L_1 $ and $ L_2 $ we choose $ B_1 $ and $ B_2 $ such that
\[
\begin{array}{rl}
L_1&=H_1^*\log_q(B_1) H_1=\varepsilon\log_q(B_1)\\[1ex]
L_2&=H_2^*\log_q(B_2) H_2=(1-\varepsilon)\log_q(B_2).
\end{array}
\]
It follows that
\[
B_1=\exp_q(\varepsilon^{-1}L_1)\qquad\text{and}\qquad B_2=\exp_q((1-\varepsilon)^{-1}L_2).
\]
Inserting in the inequality we now obtain
\[
\begin{array}{l}
\tr\exp_q(L_1+L_2)\\[1.5ex]
\le\tr\exp_q(L_1)^{2-q}\bigl(\varepsilon \exp_q(\varepsilon^{-1}L_1)^{q-1}+(1-\varepsilon)\exp_q((1-\varepsilon)^{-1}L_2)\bigr)\\[1.5ex]
=\tr\exp_q(L_1)^{2-q}\bigl(L_1(q-1)+\varepsilon+(1-\varepsilon)\exp_q((1-\varepsilon)^{-1}L_2)\bigr).
\end{array}
\]
This expression decouble $ L_1 $ and $ L_2 $ and reduces the minimisation problem over $ \varepsilon $ to the commutative case. We furthermore realise that minimum is obtained by letting $ \varepsilon $ tend to zero and that
\[
\lim_{\varepsilon\to 0} (1-\varepsilon)\exp_q\bigl((1-\varepsilon)^{-1}L_2\bigr)=\exp_q(L_2).
\]
We finally replace $ L_1 $ and $ L_2 $ with $ A $ and $ B. $ This proves the first statement in Theorem~\ref{main theorem}.\\[1ex]
The proof of the second statement is virtually identical to the proof of the first. Since now $ 2\le q\le 3 $ the second inequality in Corollary~\ref{main inequality for k variables} applies. Setting $ k=2 $ and applying the same substitutions as in the proof of the first statement we arrive at the inequality
\[
\begin{array}{l}
\tr\exp_q(L_1+L_2)\\[1.5ex]
\ge \tr\exp_q(L_1)^{2-q}\bigl(L_1(q-1)+\varepsilon+(1-\varepsilon)\exp_q((1-\varepsilon)^{-1}L_2)\bigr).
\end{array}
\]
Since $ 2\le q\le 3 $ the function
\[
\varepsilon\to\varepsilon+(1-\varepsilon)\exp_q((1-\varepsilon)^{-1}L_2)
\]
is now decreasing, and we thus maximise the right hand side in the above inequality by letting $ \varepsilon $ tend to zero. This proves the second statement in Theorem~\ref{main theorem}.

{\small


\vfill

\noindent Frank Hansen: Institute for Excellence in Higher Education, Tohoku University, Japan.\\
Email: frank.hansen@m.tohoku.ac.jp.
      }

\end{document}